\documentclass[preprint,11pt,3p]{elsarticle}

\usepackage{fixltx2e,graphicx,breakurl}
\usepackage{amsmath,amssymb,amsthm}
\usepackage[mathlines,displaymath]{lineno}

\makeatletter
\def\Ddots{\mathinner{\mkern1mu\raise\p@
\vbox{\kern7\p@\hbox{.}}\mkern2mu
\raise4\p@\hbox{.}\mkern2mu\raise7\p@\hbox{.}\mkern1mu}}
\makeatother

\newcommand{\nats}{{\mathbb N}}

\renewcommand{\alph}{\textrm{Alph}}
\newcommand{\Pal}{\textrm{PAL}}
\newcommand{\Fac}{\textrm{Fac}}
\newcommand{\MP}{\textit{MinPal}}
\renewcommand{\epsilon}{\varepsilon}

\newtheorem{thm}{Theorem}

\newtheorem{theorem}{Theorem}[section]

\newtheorem{lemma}[theorem]{Lemma}
\newtheorem{corollary}[theorem]{Corollary}
\newtheorem{proposition}[theorem]{Proposition}

\newtheorem{remark}[theorem]{Remark}

\newtheorem{claim}[theorem]{Claim}

\theoremstyle{definition}

\begin{document}

\begin{frontmatter}

\title{\textbf{On the Least Number of Palindromes\\ Contained in an Infinite Word}}

\author{Gabriele Fici\corref{cor1}}
\ead{gabriele.fici@unipa.it}
% \ead[url]{home page}
\address{Dipartimento di Matematica e Informatica, Universit\`a di Palermo, Italy}

\author{Luca Q. Zamboni}
 \ead{lupastis@gmail.com}
% \ead[url]{home page}
 \address{Universit\'e Claude Bernard Lyon 1, France, and University of Turku, Finland}

\cortext[cor1]{Corresponding author.}

\journal{Theoretical Computer Science}

\begin{abstract}
We investigate the least number of palindromic factors in an infinite word. We first consider general alphabets, and give answers to this problem for periodic and non-periodic words, closed or not under reversal of factors. We then investigate the same problem when the alphabet has size two.
\end{abstract}


\end{frontmatter}

%\linenumbers

%%%%%%%%%%%%%%%%%%%%%%%%%%%%%%%%%%%%%%%%%%%%%%%%%%%%%%%%%%%%%%%%%%%%%%%%%%%%%%%%%%%%%%%%%%%%%%
\section{Introduction}
%%%%%%%%%%%%%%%%%%%%%%%%%%%%%%%%%%%%%%%%%%%%%%%%%%%%%%%%%%%%%%%%%%%%%%%%%%%%%%%%%%%%%%%%%%%%%%

In recent years, there has been an increasing interest in the importance of palindromes in mathematics, theoretical computer science and theoretical physics. In particular, one is interested in infinite words containing arbitrarily long palindromes. This stems from their role in the modeling of {\em quasicrystals} in theoretical physics (see for instance~\cite{dDlZ03comb,aHoKbS95sing}) and also diophantine approximation in number theory (e.g., see~\cite{sF06pali2,AB1,AB2,AB3,AA,BL,Roy1,Roy2}).

In~\cite{xDjJgP01epis},  X.~Droubay, J. Justin and G.~Pirillo observed that any finite word $w$ of {\em length} $|w|$  contains at most $|w| + 1$ distinct palindromes (including the empty word).  Such words  are `rich' in palindromes, in the sense that they contain the maximum number of different palindromic factors. Accordingly, we say that a finite word $w$ is {\it rich} if it contains exactly $|w| + 1$ distinct palindromes, and we call an infinite word {\it rich} if all of its factors are rich. In an independent work, P.~Ambro{\v{z}}, C.~Frougny, Z.~Mas{\'a}kov{\'a} and E.~Pelantov{\'a}~\cite{pAzMePcF06pali}  have considered the same class of words, which they call {\it full words} (following earlier work of S.~Brlek, S.~Hamel, M.~Nivat, and C.~Reutenauer~\cite{Brl}). Since~\cite{xDjJgP01epis}, there is an extensive number of papers devoted to the study of rich words and their generalizations (see~\cite{jAmBjCdD03pali,pAzMePcF06pali,pBzMeP07fact,Brl,BDGZ1,BDGZ2,DLGZ,GJWZ}).

In this note we consider the opposite question: What is the least number of palindromes which can occur in an infinite word subject to certain constraints? For an infinite word $\omega$, the set $\Pal(\omega)$ of palindromic factors of $\omega$ can be finite or infinite (cf.~\cite{Brl}). For instance, in case $\omega$ is a Sturmian word, then $\Pal(\omega)$ contains two elements of each odd length and one element of each even length. In fact, this property characterizes Sturmian words (see~\cite{aD97stur} and \cite{xDgP99pali}). In contrast, the {\it paperfolding word} $P$ is an example of an aperiodic uniformly recurrent  word  containing a finite number of palindromes: 
\[P=aabaabbaaabbabbaaabaabbbaabbabbaaaba\cdots\] It is obtained as the limit of the sequence $(P_{n})_{n\ge 0}$ defined recursively by $P_{0}=a$ and $P_{n+1}=P_{n}a\hat{P}_{n}$ (for $n\geq 0)$,  where $\hat{P}_n$ is the word obtained from $\tilde{P}_n$ by exchanging $a$'s and $b$'s \cite{All}. J.-P. Allouche showed that the paperfolding word contains exactly $29$ palindromes, the longest of which has length $13.$

It is easy to see that any uniformly recurrent word $\omega $ which contains an infinite number of palindromes must be {\it closed under reversal}, that is, for every factor $u=u_1u_2\cdots u_n$ of $\omega,$ its {\it reversal}  $\tilde{u}=u_n\cdots u_2u_1$ is also a factor of $\omega.$ The converse is not true: In fact, J. Berstel, L. Boasson, O. Carton and I. Fagnot~\cite{Ber} exhibited various examples of uniformly recurrent words closed under reversal and containing a finite number of palindromes. 
The paperfolding word is not closed under reversal, since for example it contains the factor $aaaba$ but not  $abaaa.$

If $X$ is a set consisting of infinite words, we set
\[\MP (X)=\textrm{inf}\{\#\Pal (\omega)\mid \omega \in X\}.\] 
 
We first show that without restrictions on the cardinality of the alphabet, one has that $\MP=4$. That is, for 
\[W=\{\omega \in A^\nats \mid 0<\#A<\infty\},\]
we have $\MP(W)=4.$ If in addition one requires that the word be aperiodic, that is, for
\[W_{\textrm{ap}}=\{\omega \in A^\nats \mid 0<\#A<\infty\,\,\mbox{and}\,\,\omega \,\,\mbox{is aperiodic}\},\]
 then $\MP(W_{\textrm{ap}})=5.$ If moreover one requires that the word must be closed under reversal, that is, for
 \[W_{\textrm{cl}}=\{\omega \in A^\nats \mid 0<\#A<\infty\,\,\mbox{and}\,\,\omega \,\,\mbox{is closed under reversal}\},\] then one still has  $\MP(W_{\textrm{cl}})=5.$\\

\noindent In the case of binary words we show the following:

\begin{thm} Let $A$ be a set with $\#A=2.$  Then:
\begin{enumerate}
\item $\MP(A^\nats)=9$, where $A^\nats$ denotes the set of all infinite words on $A.$
\item $\MP(A_{\textrm{ap}}^\nats)=11$, where $A_{\textrm{ap}}^\nats$ denotes the set of all aperiodic words in $A^\nats.$
\item $\MP(A_{\textrm{cl}}^\nats)=13$, where  $A_{\textrm{cl}}^\nats$ denotes the set of all words in $A^\nats$ closed under reversal.
\item $\MP(A_{\textrm{ap/cl}}^\nats)=13$, where $A_{\textrm{ap/cl}}^\nats$ denotes the set of all aperiodic words in $A^\nats$ closed under reversal.
\end{enumerate}
\end{thm}

%%%%%%%%%%%%%%%%%%%%%%%%%%%%%%%%%%%%%%%%%%%%%%%%%%%%%%%%%%%%%%%%%%%%%%%%%%%%%%%%%%%%%%%%%%%%%%
\section{Definitions and Notations}
%%%%%%%%%%%%%%%%%%%%%%%%%%%%%%%%%%%%%%%%%%%%%%%%%%%%%%%%%%%%%%%%%%%%%%%%%%%%%%%%%%%%%%%%%%%%%%

Given a finite non-empty set $A$ (called the {\it alphabet}), we denote by $A^*$  and $A^\nats$   respectively the set of finite words and the set of (right) infinite words over the alphabet $A$. Given a finite word $u =a_1a_2\cdots a_n$ with $n \geq 1$ and $a_i \in A,$ we denote the length $n$ of $u$ by $|u|.$ The  \textit{empty word} will be denoted by $\varepsilon$ and we set $|\varepsilon|=0.$ We put $ A^+= A^*-\{\varepsilon\}.$  For $u,v \in A^+$ we denote by  $|u|_v$ the number of occurrences of $v$ in $u.$ For instance $|0110010|_{01}=2.$ We denote the reverse of $u$ by $\tilde{u},$ i.e., $\tilde{u}=a_n\cdots a_2a_1.$ 

Given a finite or  infinite word $\omega =\omega_0\omega_1\omega_2\cdots $ with $\omega_i\in A,$ we say a word $u\in A^+$ is a  {\it factor} of $\omega$ if  $u=\omega_{i}\omega_{i+1}\cdots \omega_{i+n}$ for some natural numbers $i$ and $n.$
We denote by $\Fac(\omega)$ the set of all factors of $\omega,$ 
and by $\alph(\omega)$ the set of all factors of $\omega$ of length $1.$ 
Given (non-empty) factors $u$ and $v$ of $\omega,$ we say $u$ is a {\it first return} to $v$ {\it in} $\omega$ if $uv$ is a factor of $\omega$ which begins and ends in $v$ and $|uv|_v=2.$ If $u$ is a first return to $v$ in $\omega$ then $uv$ is called a {\it complete first return} to $v$ in $\omega$.

A factor $u$ of $\omega$ is called {\it right special} if both $ua$ and $ub$ are factors of $\omega$ for some pair of distinct letters $a,b \in A.$ Similarly, $u$ is called {\it left special} if both $au$ and $bu$ are factors of $\omega$ for some pair of distinct letters $a,b \in A.$ The factor $u$ is called $bispecial$ if it is both right special and left special.
For each factor $u$ of $\omega$, we  set
\[\omega\big|_{u}=\{ n\in \nats \mid \omega_n\omega_{n+1}\cdots \omega_{n+|u|-1}=u\}.\]
We say $\omega$ is {\it recurrent} if for every $u\in \Fac(\omega)$  the set $\omega\big|_u$ is infinite.
We say $\omega$ is {\it uniformly recurrent} if for every $u\in \Fac(\omega)$  the set $\omega\big|_u$ is syndedic, i.e., of bounded gap. A word $\omega\in A^\nats$  is {\it (purely) periodic} if there exists a positive integer $p$ such that $\omega_{i+p} = \omega_i$ for all indices $i$, and it is {\it ultimately periodic} if $\omega_{i+p} = \omega_i$ for all sufficiently large $i$. For a finite word $u=a_1a_2\cdots a_n$, we call $p$ a \emph{period} of $u$ if $a_{i+p}=a_{i}$ for every $1 \le i \le n-p$, and we denote by $\pi(u)$ the least period of $u$. Finally, a word $\omega\in A^\nats$ is called {\it aperiodic} if it is not ultimately periodic. 
Two finite or infinite words are said to be {\it isomorphic} if the two words are equal up to a renaming of the  letters. We denote by $[\omega]$ the set of words that are isomorphic to $\omega $ or to $\tilde{\omega}$. Note that any word in $[\omega]$ has the same periods as $\omega$.

We denote by $\Pal(\omega)$ the set of all palindromic factors of $\omega,$ i.e., the set of all factors $u$ of $\omega$ with $\tilde{u}=u.$
We have that $\Pal(\omega)$ contains at least $\varepsilon$ and $\alph(w)$.

%%%%%%%%%%%%%%%%%%%%%%%%%%%%%%%%%%%%%%%%%%%%%%%%%%%%%%%%%%%%%%%%%%%%%%%%%%%%%%%%%%%%%%%%%%%%%%
\section{General alphabets}
%%%%%%%%%%%%%%%%%%%%%%%%%%%%%%%%%%%%%%%%%%%%%%%%%%%%%%%%%%%%%%%%%%%%%%%%%%%%%%%%%%%%%%%%%%%%%%

The following lemma follows from a direct inspection.

\begin{lemma}\label{lem:rich}
Every word $w$ of length $9$ such that $\alph(w)=2$ contains at least $9$ palindromes.
\end{lemma}

An application of the previous lemma is the following.

\begin{proposition}
 Every infinite word contains at least $4$ palindromes.
\end{proposition}

\begin{proof}
 The empty word and the letters are palindromes. Therefore, if an infinite word $\omega$ contains only $3$ palindromes then $\alph(\omega)=2$. This is in contradiction with Lemma \ref{lem:rich}. 
\end{proof}

\noindent We have the following characterization of words containing only $4$ palindromes.

\begin{proposition}\label{prop:4}
If an infinite word contains exactly $4$ palindromes, then it is of the form $u^{\infty}$ where $u$ is of the form $u=abc$ with $a,b,$ and $c$ distinct letters. 
\end{proposition}

\begin{proof}
Let $\omega$ be an infinite word containing exactly $4$ palindromes. By Lemma \ref{lem:rich}, it follows that  $\alph(\omega) = 3$.  Should $\omega$ contain a factor of the form $aa$ or of the form $aba$, then $\omega$ would contain at least $5$ palindromes. The statement now follows.
\end{proof}

\begin{corollary}
 Every non-periodic infinite word contains at least $5$ palindromes.
\end{corollary}

In fact, there exist non-periodic uniformly recurrent words containing only $5$ palindromes. Let $F$ be the Fibonacci word, that is the word \[F=abaababaabaababaababaabaab\cdots\]
obtained as the limit of the sequence $(f_{n})_{n\ge 0}$, where $f_{0}=b$, $f_{1}=a$ and $f_{n+1}=f_{n}f_{n-1}$. The image of the Fibonacci word $F$ under the morphism $\phi: a\mapsto a, b\mapsto bc$,
\[\phi(F)=abcaabcabcaabcaabcabcaabcabca \cdots\] 
contains only 5 palindromes, namely: $\epsilon$, $a$, $b$, $c$ and $aa$. Note that the word $\phi(F)$ is not closed under reversal, since for example it does not contain the reversal of the factor $bc$. 

Berstel et al.~\cite{Ber} exhibited a uniformly recurrent word over a four-letter alphabet closed under reversal and containing only $5$ palindromes (the letters and the empty word):
\[\omega=abcdbacdabdcbacdabcdbadcabdcba\cdots\] defined as the limit of the sequence $(U_{n})_{n\ge 0}$, where $U_{0}=ab$ and $U_{n+1}=U_{n}cd\tilde{U}_{n}$.

%%%%%%%%%%%%%%%%%%%%%%%%%%%%%%%%%%%%%%%%%%%%%%%%%%%%%%%%%%%%%%%%%%%%%%%%%%%%%%%%%%%%%%%%%%%%%%
\section{Binary alphabet}
%%%%%%%%%%%%%%%%%%%%%%%%%%%%%%%%%%%%%%%%%%%%%%%%%%%%%%%%%%%%%%%%%%%%%%%%%%%%%%%%%%%%%%%%%%%%%%

In this section we fix a binary alphabet $A=\{a,b\}$.
As a consequence of Lemma \ref{lem:rich}, every infinite word over $A$ contains at least $9$ palindromes. 

By direct computation, if $w$ is a word over $A$ of length $12,$ then $\#\Pal(w)\geq 9$ and $ \#\Pal(w)=9$ if and only if $w=u^2$ where $u\in [v]$ and $v=aababb.$  Indeed, for each $u\in [v]$ one has \[\Pal(u^{2})=\{\epsilon,a,b,aa,bb,aba,bab,abba,baab\}.\]
Since no palindrome of length $5$ or $6$ occurs in $u^{2}$, the word $u^{\infty}$ contains only $9$ palindromes. Moreover, for each  $u\in [v]$ and $\alpha \in A$, if $\pi(u^{2}\alpha)\neq 6$, then $u^{2}\alpha$ contains at least $10$ palindromes. So we have:

\begin{proposition}\label{prop:9per}
Let $v=aababb$. An infinite word over $A$ contains exactly $9$ palindromes if and only if it is of the form $u^{\infty}$ for some $u\in [v]$. In particular it is  periodic  of period $6$.
\end{proposition}

\noindent We next characterize all binary words containing precisely $10$ palindromes. 
By direct inspection, any word over $A$ of length $14$ containing precisely $10$ palindromes belongs to one of the  following four sets:
\begin{enumerate}
 \item $T_{1}=\{w^{2} \mid w\in [av]\}$;
 \item $T_{2}=\{w^{2} \mid w\in [vb]\}$;
 \item $T_{3}=\{\alpha w^{2}\beta \mid \alpha,\beta \in A, w\in [v], \pi(\alpha w^{2})\neq 6, \pi(w^{2}\beta)=6\}$;
 \item $T_{4}=\{w^{2}\alpha \beta \mid \alpha,\beta \in A, w\in [v], \pi(w^{2}\alpha)=6, \pi(w^{2}\alpha \beta)\neq 6\}$.
\end{enumerate}

Moreover, the length of the longest palindrome in any of the words in the sets $T_{i}$ is at most $6$.

\begin{lemma}
 Let $\gamma\in A$. Then:
\begin{enumerate}
 \item if $w^{2}\in T_{1}$ and $\pi(w^{2}\gamma)\neq 7$, then $w^{2}\gamma$ contains $11$ palindromes;
 \item if $w^{2}\in T_{2}$ and $\pi(w^{2}\gamma)\neq 7$, then $w^{2}\gamma$ contains $11$ palindromes; 
 \item if $\alpha w^{2}\beta\in T_{3}$ and $\pi(w^{2}\beta\gamma)\neq 6$, then $w^{2}\beta\gamma$ contains $11$ palindromes;
  \item if $w^{2}\alpha \beta\in T_{4}$, then $w^{2}\alpha\beta\gamma$ contains $11$ palindromes.
\end{enumerate}
\end{lemma}

\noindent Thus we have:

\begin{proposition}\label{prop:10per}
An infinite word $w$ over $A$ contains exactly $10$ palindromes if and only if $w$ is of the form $u^{\infty}$ with $u\in [av]$ or $u\in [vb]$, or of the form $\alpha (u)^{\infty}$ with $u\in [v]$ and  $\alpha \in A$ such that $\alpha u$ does not have period $6.$ In the first case $w$ is  periodic of period $7,$ while in the second case $w$ is  ultimately periodic of period $6$.
\end{proposition}

Thus, every aperiodic word over $A$ contains at least $11$ palindromes. An example of a uniformly recurrent aperiodic word  containing exactly $11$ palindromes is the image of the Fibonacci word $F$ under the morphism $\psi: a\mapsto a, b\mapsto abbab$,
 \[\psi(F)=aabbabaaabbabaabbabaaabbabaaabbab\, \cdots\]
 The palindromes in $\psi(F)$ are: $\epsilon$, $a$, $b$, $aa$, $bb$, $aaa$, $aba$, $bab$, $abba$, $baab$ and $baaab$.
%The palindromes of $\psi(F)$ are the palindromes of $w$ and the words $aaa$ and $baaab$.  
Note that the word $\psi(F)$ is not closed under reversal, since, for example, it does not contain the reversal of its factor $abaaa$.

Berstel et al.~\cite{Ber} exhibited an aperiodic uniformly recurrent word closed under reversal and containing precisely $17$ palindromes. It is the word obtained from the paperfolding word $P$ by applying the morphism $\tau:a \mapsto ab, b\mapsto ba$: \[\tau(P)=ababbaababbabaabababbabaabbabaababab\cdots\] 

In the next section, we show that the least number of palindromes which can occur in an infinite binary word closed under reversal is $13$.

Rather than limiting the total number of palindromic factors, we consider the problem of limiting the length of the longest palindromic factor. In this case we have:

\begin{proposition}\label{5&6} Every infinite binary word contains a palindromic factor of length  greater than $3.$ There exist infinite binary words containing no palindromic factor of length greater than $4$,  but every such word is ultimately periodic.  There exists a uniformly recurrent aperiodic binary word  (closed under reversal) whose longest palindromic factor has length $5.$\end{proposition}

\begin{proof} Let $\omega \in \{a,b\}^\nats.$ We will show that $\omega$ contains a palindromic factor of length greater than $3$. Let $T$ denote the shift map, that is, $T\omega$ is the word whose $i$-th letter is $\omega_{i+1}$. If $aaa\in \Fac(T\omega),$ then 
\[\{aaaa, baaab\}\cap \Fac(\omega)\neq \emptyset.\] 
Thus we can assume that neither $aaa$ nor $bbb$ is a factor of $T\omega.$ If neither $aa$ nor $bb$ occurs in $T\omega,$ then $T\omega=(ab)^\infty$ or $T\omega=(ba)^\infty.$  In either case $T\omega$ contains the palindrome $ababa.$ Thus, without loss of generality we can assume that  $aa$ occurs in $T\omega.$ If we now consider all the possible right extensions of $aa$ which avoid $aaa$ and $bbb$, we find that each terminates in a palindrome of length $4$ or $5$:
\[
\begin{cases}
aabaa\\
aababa\\
aababba\\
aabba
\end{cases}\]

Next, suppose $\omega$ contains no palindromic factor of length greater than $4.$ We will show that $\omega$ is ultimately periodic and give an example of such a word. If $T\omega$ contains $aaaa,$ then $baaaa \in \Fac(\omega)$  which implies that
\[\{aaaaa, baaaab\}\cap \Fac(\omega)\neq \emptyset,\]
a contradiction. Thus we can assume that $aaaa \notin \Fac(T\omega).$ If $aaa\in \Fac(T^2\omega),$ then $baaa\in  \Fac(T\omega)$, which implies that $baaab\in \Fac(\omega),$ a contradiction. Thus neither $aaa$ nor $bbb$ occurs in $T^2\omega.$ If neither $aa$ nor $bb$ occurs in $T^2\omega$, we would have that $T^2\omega=(ab)^\infty$ or $T^2\omega=(ba)^\infty,$ a contradiction since each contains $ababa.$ Thus, without loss of generality we can assume that $aab$ occurs in $T^2\omega.$ It is readily verified that the only possible first returns to $aab$ in $\omega$ are 
\[
\begin{cases}
aababb\\
aabbab\\
\end{cases}\]
If $aabbab$ occurs in $\omega,$ then $\omega$ has a tail of the form $(aabbab)^\infty.$ If $aabbab$ does not occur in $\omega,$ the $\omega$ has a tail of the form $(aababb)^\infty.$ In either case $\omega$ is ultimately periodic.
It is readily verified that $(aabbab)^\infty$ has no palindromic factor of length greater than $4.$ 

Finally, we show the existence of a uniformly recurrent aperiodic binary word $\omega$  (closed under reversal) whose longest palindromic factor has length $5.$
Set $U_0=aabb$ and for $n\geq 0,$
\[
\begin{cases}
U_{2n+1}=U_{2n} ab \tilde{U}_{2n};\\

U_{2n}=U_{2n-1} ba \tilde{U}_{2n-1}.
\end{cases}\] 
Then $U_n$ is a prefix of $U_{n+1}$  for each $n\geq 0$ and we set
\[\omega = \lim_{n\rightarrow \infty}U_n.\]
Then, by construction, $\omega$ is closed under reversal and is uniformly recurrent (in fact, the recursive definition of $\omega$ shows that each prefix $U_n$ occurs in $\omega$ with bounded gap).
Now, a straightforward verification shows that
\[\Pal(U_2)=\{\varepsilon, a, b, aa, bb, aaa, aba, bab, bbb, abba, baab, aabaa, abbba, baaab, bbabb\}.\]
We note that $\Pal(\tilde{U}_n)=\Pal(U_n)$ for each $n\geq 0.$
We prove by induction on $n\geq 2$ that no other palindrome occurs in $U_n.$ From the above equality, we have that the result holds for $n=2.$ Now, suppose that $n\geq 2$ and  $\#\Pal(U_n)=15.$ We will show that  $\#\Pal(U_{n+1})=15.$  For $n\geq 2, $ we can write $U_n=U_1t_n\tilde{U_1}$ and 
\[U_{n+1}=\begin{cases}
U_1t_n\tilde{U_1}abU_1\tilde{t}_n\tilde{U}_1&\text{for $n$ even;}\\
U_1t_n\tilde{U_1}baU_1\tilde{t}_n\tilde{U}_1&\text{for $n$ odd.}
\end{cases}\]
Considering that $|U_1|=10,$ if $U_{n+1}$ contained a palindrome $v$ of length $6$ or $7,$ then either $v$ would be contained in $U_n$ or in $U_2.$ Thus, $\#\Pal(U_{n+1})=15.$
Finally, by Lemma~\ref{recurrent}, we deduce that $\omega$ is aperiodic. 
\end{proof}

%%%%%%%%%%%%%%%%%%%%%%%%%%%%%%%%%%%%%%%%%%%%%%%%%%%%%%%%%%%%%%%%%%%%%%%%%%%%%%%%%%%%%%%%%%%%%%
\section{The case of binary words closed under reversal}
%%%%%%%%%%%%%%%%%%%%%%%%%%%%%%%%%%%%%%%%%%%%%%%%%%%%%%%%%%%%%%%%%%%%%%%%%%%%%%%%%%%%%%%%%%%%%%

In this section we will prove the following:

\begin{theorem}\label{13} Let $A=\{a,b\},$ and let $A_{\textrm{cl}}^\nats$ (respectively, $A_{\textrm{ab/cl}}^\nats)$ denote the set of all infinite words in $A^\nats$ closed under reversal (respectively, the set of all aperiodic words in $A^\nats$ closed under reversal). Then $\MP(A_{\textrm{cl}}^\nats)=\MP(A_{\textrm{ab/cl}}^\nats)=13.$
\end{theorem}

Note that since there exist aperiodic binary words closed under reversal containing a finite number of palindromic factors (see for instance \cite{Ber} or Lemma~\ref{finite}), we have that \[\MP(A_{\textrm{cl}}^\nats)\leq \MP(A_{\textrm{ab/cl}}^\nats)<+\infty.\]

Our proof will involve some intermediate lemmas and a case-by-case analysis.
We begin with some general remarks concerning words closed under reversal. The following lemma is probably well known but we include it here for the sake of completeness:

\begin{lemma}\label{recurrent} Suppose   $\omega \in A^\nats$ is closed under reversal. Then $\omega$ is recurrent. Hence $\omega$ is either aperiodic or (purely) periodic. In the latter case $\#\Pal(w)=+\infty.$
\end{lemma}

\begin{proof} Let $u$ be a prefix of $\omega.$  Then $\tilde{u}$ occurs in $\omega$ followed by some letter $\alpha\in A.$ Then $\alpha u$ is a factor of $\omega$, which means that $u$ occurs at least twice in $\omega.$ This proves that $\omega$ is recurrent.
If $\omega$ is ultimately periodic, meaning $\omega=vu^\infty$ for some $u$ and $v,$ then as $\omega$ is recurrent it follows that $u$ is a factor of $v^\infty$, which implies that $\omega$ is purely periodic.  Finally, it remains to prove that if $\omega=u^\infty$ for some factor $u$ of $\omega,$ then $\omega$ contains an infinite number of palindromic factors, or equivalently that for each $M$ there exists a palindromic factor $v$ of $\omega$ with $|v|\geq M.$ So let $M$ be given and pick a positive integer $n$ such that $n|u|\geq M.$ Since $u^n$ is a factor of $\omega$ it follows that  $\tilde{u}^n$ is a factor of $\omega.$ This implies that $\tilde{u}^n$ occurs in $u^{n+1}.$ Thus there exists a factor $v$ of $\omega$ with $n|u|\leq |v|<(n+1)|u|$ which begins in $u^n$ and ends in $\tilde{u}^n.$ Hence $v$ is a palindrome. 
\end{proof}

\begin{corollary} $\MP(A_{\textrm{cl}}^\nats)=\MP(A_{\textrm{ab/cl}}^\nats).$
\end{corollary}

\begin{proof} We already noticed that $\MP(A_{\textrm{cl}}^\nats) \leq \MP(A_{\textrm{ab/cl}}^\nats)<+\infty.$ Let $\omega \in A_{\textrm{cl}}^\nats$ with $\#\Pal(\omega)=\MP(A_{\textrm{cl}}^\nats).$ Since $\MP(A_{\textrm{cl}}^\nats)<+\infty ,$ it follows from Lemma~\ref{recurrent} that $\omega$ is aperiodic and hence $\omega \in A_{\textrm{ab/cl}}^\nats, $ whence $\MP(A_{\textrm{ab/cl}}^\nats)\leq \#\Pal(\omega)=\MP(A_{\textrm{cl}}^\nats).$
\end{proof}

 We begin by showing that $13$ is an upper bound for $\MP(A_{\textrm{cl}}^\nats).$

\begin{lemma}\label{finite} $\MP(A_{\textrm{cl}}^\nats)\leq 13.$
\end{lemma}

\begin{proof} Set $U_0=abaabbabaaabbaaba.$ For $n\geq 0$ define
\[
\begin{cases}U_{2n+1}=U_{2n} \,bbaa \,\tilde{U}_{2n};\\ 
U_{2n+2}=U_{2n+1}\,aabb\,\tilde{U}_{2n+1}.
\end{cases}\] 
The first few values of $U_n$ are as follows: 
\begin{align*}
U_1&=abaabbabaaabbaaba \,bbaa\, abaabbaaababbaaba\\
U_2&=abaabbabaaabbaababbaaabaabbaaababbaaba \,aabb\, abaabbabaaabbaabaaabbabaabbaaababbaaba.
\end{align*}
We note that $U_n$ is a prefix of $U_{n+1}$ for each $n\geq 0.$ 
Let
\[\omega=\lim_{n\rightarrow +\infty} U_n.\]
It is clear by construction that $\omega$ is closed under reversal. It is readily verified that 
\[\Pal(U_2)=\{\epsilon, a, aa, aaa, aabaa, aabbaa, aba, abba, b, baaab, baab, bab, bb\}.\]
Hence $\#\Pal(U_2)=13.$

We prove by induction on $n\geq 2$ that  $\#\Pal(U_n)=13.$ The above equality shows that this is true for $n=2$. For every $n\ge 2$, we can write $U_n=U_1t_n\tilde{U_1}$ and 
\[U_{n+1}=\begin{cases}
U_1t_n\tilde{U_1}bbaaU_1\tilde{t}_n\tilde{U}_1&\text{for $n$ even;}\\
U_1t_n\tilde{U_1}aabbU_1\tilde{t}_n\tilde{U}_1&\text{for $n$ odd.}
\end{cases}\]
Since $|U_{1}|=38,$ a palindrome of length smaller than or equal to $8$ which occurs in $U_{n+1}$ must either occur in $U_{n}$ or in $U_2.$ The result now follows from the induction hypothesis.
\end{proof}

Set
\[\Omega =\{\omega \in \{a,b\}\mid \omega\,\mbox{is closed under reversal and}\,\#\Pal(\omega)=\MP(A_{\textrm{cl}}^\nats)\}.\]
For $\omega \in \Omega$ and $x\in \{a,b\}$ we put
\[N_x(\omega)=\textrm{max}\{k\mid x^k\,\mbox{is a factor of}\,\omega\}.\]
Since $\MP(A_{\textrm{cl}}^\nats) <+\infty$ both $N_a$ and $N_b$ are finite.

\begin{proof}[Proof of Theorem~\ref{13}] Fix $\omega \in \Omega.$  We must show that $\#\Pal(\omega)\geq 13.$  We will make use of the following lemma:

\begin{lemma}\label{bab} Let $k$ be a positive integer. If $a^k$ (respectively $b^k$) is a factor of $\omega,$ then so is $ba^kb$ (respectively $ab^ka$).
\end{lemma}

\begin{proof} Suppose to the contrary that $a^k$ is a factor of $\omega$ but not $ba^kb.$  By Lemma~\ref{recurrent}, $\omega$ is recurrent and hence $ba^{N_a(\omega)}b$ is a factor of $\omega.$ Hence $ 1\leq k<N_a(\omega).$ 
Let $\omega'$ be a tail of $\omega$ beginning in $b.$ Let $\nu\in \{a,b\}^\nats$ be the word obtained from $\omega'$ by
replacing all occurrences of $ba^jb$ in $\omega'$  by $ba^{j-1}b$  for each $k+1\leq j\leq N_a(\omega).$ 
Thus $N_a(\nu)=N_a(\omega)-1.$ It is readily verified that $\nu$ is closed under reversal. Moreover, to every palindrome $v$ in $\nu$ corresponds a unique palindrome $w$ in $\omega$ obtained from $v$ by increasing the $a$ runs in $v$ of length $\geq k$ by one unit and leaving the other $a$ runs the same. This defines an injection
$\phi: \Pal(\nu) \hookrightarrow \Pal(\omega)$ which is not a surjection since in particular $a^k$ is not in the image of $\phi.$ Thus  $\#\Pal(\nu)<\#\Pal(\omega),$ contradicting that $\omega$ had the least number of palindromic factors amongst all binary words closed under reversal.
\end{proof} 

\begin{lemma}\label{asquare} If $\#\Pal(\omega)\leq 12$, then $a^2$ and $b^2$ are factors of $\omega$.
\end{lemma}

\begin{proof} Since $\#\Pal(\omega)\leq 12$, no factor of $\omega$ of length $12$ is rich. There are $850$ binary non-rich words $u$ of length $12.$ For each such  $u$ we compute $\Pal(u).$ By computer verification, we observe that the only cases when $\Pal(u)$ does not contain both $a^2$ and $b^2$ is when  $\Pal(u)$
is equal to one of the following $4$ sets:
\begin{align*}
A&=\{\varepsilon , a, aba, abba, abbba, b, bab, babbab, babbbab, bb, bbabb, bbb\};\\
B&=\{\varepsilon, a, aba, abba, b, bab, babab, babbab, bb, bbababb, bbabb, bbb\};\\
C&=\{\varepsilon, b, bab, baab, baaab, a, aba, abaaba, abaaaba, aa, aabaa, aaa\}; \\
D&=\{\varepsilon, b, bab, baab, a, aba, ababa, abaaba, aa, aababaa, aabaa, aaa\}.
\end{align*}
For instance,  if $u=aaababaabaaa$ then $\Pal(u)=D.$ So either $\Pal(\omega)$ contains both $a^2$ and $b^2$ or $\Pal(\omega)$ contains one of $A,B,C,$ or $D.$ Note that each of the above sets is of cardinality $12.$ By Lemma~\ref{bab}, if $B\subseteq \Pal(\omega),$ then $\Pal(\omega)$ also contains $abbba $ and hence $\#\Pal(\omega)\geq 13,$ a contradiction. A similar argument shows that  $\Pal(\omega)$ cannot contain $D.$ 
Next, suppose $A\subseteq \Pal(\omega)$ and let us consider the possible right extensions of the palindrome $babbbab$ which avoid $a^2.$ We will put a dot $(.)$ to designate positions of choice. They are: 
$babbbab.ab$ (which yields a $13$th palindrome $babab),$ $babbbab.b.ab.ab$ (which yields a $13$th palindrome $babab),$
$babbbab.b.ab.b$ (which yields a $13$th palindrome $bbabbabb),$ and finally $babbbab.b.b$ (which yields a $13$th palindrome $bbbabbb).$ In either case $\#\Pal(\omega)\geq 13,$ a contradiction. A similar argument shows that $\Pal(\omega)$ cannot contain $C.$ Thus, if $\#\Pal(\omega)\leq 12$ then both $a^2$ and $b^2$ are factors of $\omega.$
\end{proof}

\noindent In view of Lemma~\ref{asquare}, we can suppose that both $a^2$ and $b^2$ belong to $\Pal(\omega)$ and hence, by Lemma~\ref{bab}, 
\[\{\varepsilon, a,b,aa,bb,aba, bab, abba, baab\}\subseteq \Pal(\omega).\]

If $\omega$ contains no palindromic factor of length greater than $4,$ then neither $aaa$ nor $bbb$ is a factor of $\omega$ (for otherwise by Lemma~\ref{bab} either $baaab$ or $abbba$ is a factor of $\omega).$
Hence we would have 
\[\{\varepsilon, a,b,aa,bb,aba, bab, abba, baab\}=\Pal(\omega).\]
This implies that $ababb$ is a factor of $\omega.$ But then the only complete  first return to $ababb$ is 
$ababbaababb$, which would imply that $\omega$ is periodic, a contradiction.
Thus, $\omega$ must contain a palindromic factor of length $5$ or of length $6.$\\

\noindent {\it Case 1:} $\omega$ contains a palindromic factor of length $5.$ \\

\noindent Without loss of generality we can suppose \[\{aabaa, babab, baaab, aaaaa\}\cap \Pal(\omega)\neq \emptyset . \] 
\noindent {\it Case 1.1:} $aabaa \in \Pal(\omega).$  Thus $\#\Pal(\omega)\geq 10.$ By considering the possible bilateral extensions of $aabaa,$
we have \[\{aaabaaa, aaabaab, baabaab\}\cap \Fac(\omega) \neq \emptyset.\]

\noindent {\it Case 1.1.1:} $aaabaaa \in \Fac(\omega).$ This gives rise to $3$ additional palindromes: $aaabaaa, aaa, baaab$ (where $baaab$ is a consequence of $aaa$ and Lemma~\ref{bab}).  Hence $\#\Pal(\omega)\geq 13.$\\

\noindent {\it Case 1.1.2:} $aaabaab \in \Fac(\omega).$ This gives rise to $2$ additional palindromes: $aaa, baaab$ so that $\#\Pal(\omega)\geq 12.$
If  $aaaa \in \Fac(\omega),$  then $\#\Pal(\omega)\geq 13.$ So we can assume that $aaaa \notin \Fac(\omega),$ in which case $abaaab \in \Fac(\omega)$ since $aaabaab$ occurs in $\omega$ preceded by $b$ and $\omega$ is closed under reversal. We leave the following technical claim for the reader:

\begin{claim} Under the conditions of Case 1.1.2, either $\#\Pal(\omega)\geq 13,$ or every complete first return to $abaaab$ is of the form $abaaab\,(babaab)^nb\,abaaab$ for $n\geq 0.$ \end{claim} 

This implies that $aabb$ is a factor of $\omega$ but not $bbaa,$ a contradiction. \\

\noindent {\it Case 1.1.3:} $baabaab \in \Fac(\omega).$ In this case $\#\Pal(\omega)\geq 11.$ We leave the following technical claim for the reader:

\begin{claim} Under the conditions of Case 1.1.3, either $\#\Pal(\omega)\geq 13,$ or every complete first return to $baabaab$ is either of the form $baabaab\,(babaab)^n\,aab$ or of the form $baabaab\,(abbaab)^naab$ for $n\geq 1.$ \end{claim} 

Since $\omega$ is closed under reversal, both forms must actually occur. But the switch from one form to the other will produce two new palindromes: $abaabaaba$, $babaabaabab.$ In either case
$\#\Pal(\omega)\geq 13.$ This completes Case 1.1.\\

\noindent {\it Case 1.2:} $babab \in \Pal(\omega).$ Thus $\#\Pal(\omega)\geq 10.$ By considering the possible bilateral extensions of $babab$, we have \[\{abababa, abababb, bbababb\}\cap \Fac(\omega) \neq \emptyset.\]

\noindent {\it Case 1.2.1:} $abababa \in \Fac(\omega).$ This gives rise to $2$ additional palindromes: $ababa, abababa$ so that $\#\Pal(\omega)\geq 12.$ But then every bilateral extension adds a $13$th palindrome: either $aabababaa$ or $bababab.$\\

 \noindent {\it Case 1.2.2:} $abababb \in \Fac(\omega).$ This gives rise to the additional palindrome $ababa$ so that $\#\Pal(\omega)\geq 11.$ We leave the following technical claim for the reader:

\begin{claim} Under the conditions of Case 1.2.2, either $\#\Pal(\omega)\geq 13,$ or every complete first return to $ababa$ is either of the form $ababa\,(bbaaba)^n\,ba$ or of the form $ababa\,(abbaba)^nba$ for $n\geq 0.$ \end{claim} 

Since $\omega$ is closed under reversal, both forms must actually occur. But the switch from one form to the other will produce two new palindromes: $aababaa$, $baababaab.$ In either case
$\#\Pal(\omega)\geq 13.$ \\

\noindent {\it Case 1.2.3:} $bbababb \in \Fac(\omega).$ This gives rise to the additional palindrome $bbababb$, so that $\#\Pal(\omega)\geq 11.$ If $bbb \in \Fac(\omega),$ then $\Pal(\omega)$ would contain $2$ additional palindromes (namely, $bbb$ and $abbba).$ So we can assume that  $bbb \notin \Fac(\omega),$ in which case
$abbababba \in \Fac(\omega)$, so that $\#\Pal(\omega)\geq 12.$ If $abbababba $ extends on either side by $b,$ we would get the $13$th palindrome $babbab.$ Otherwise, $aabbababbaa \in \Fac(\omega).$ In either case, 
$\#\Pal(\omega)\geq 13.$ This completes Case 1.2.\\

\noindent {\it Case 1.3:} $baaab \in \Pal(\omega).$ In this case
\[\{\varepsilon, a, b, aa, bb, aba, bab, abba, baab, aaa, baaab\}\subseteq \Pal(\omega),\]
 and hence $\#\Pal(\omega)\geq 11.$ If either $aaaa$ or $bbb$ is a factor of $\omega,$ this would give rise to $2$ additional palindromes, whence $\#\Pal(\omega)\geq 13.$ So we can assume that $aaaa$ and $bbb$ are not factors
of $\omega$ and in view of cases 1.1 and 1.2, that $baaab$ is the only palindromic factor of $\omega $ of length $5.$
By considering the possible bilateral extensions of $baaab$, we have \[\{bbaaabb, abaaaba, abaaabb\}\cap \Fac(\omega) \neq \emptyset.\]\\
\noindent {\it Case 1.3.1:} $bbaaabb \in \Fac(\omega).$ But then so is $abbaaabba$ whence $\#\Pal(\omega)\geq 13.$ \\

\noindent {\it Case 1.3.2:} $abaaaba \in \Fac(\omega).$ So $\#\Pal(\omega)\geq 12.$ But any bilateral extension of $abaaaba$ adds a $13$th palindrome: either $aabaa$ or $babaaabab.$ In either case  $\#\Pal(\omega)\geq 13.$ \\
 
\noindent {\it Case 1.3.3:} $abaaabb \in \Fac(\omega).$ So $\#\Pal(\omega)\geq 11.$ We leave the following technical claim for the reader:

\begin{claim} Under the conditions of Case 1.3.3 (which include that $aaaa$ and $bbb$ are not factors
of $\omega$ and  that $baaab$ is the only palindromic factor of $\omega $ of length $5)$  either $\#\Pal(\omega)\geq 13,$ or every first return to $baaab$ is of one of $4$ types: 
\begin{itemize}
\item  $x_n=baaab(baabab)^n$ for some $n\geq 1$;
\item $y_n=baaab(babaab)^n$ for some $n\geq 1$;
\item $w_n=baaab(abbaab)^nab$ for some $n\geq 0$;
\item $z_n=baaab(babaab)^nba$ for some $n\geq 0.$
\end{itemize}
\end{claim} 

We now consider two consecutive first returns to $baaab$ in $\omega$. If any combination from the following set should occur:
\[\{x_nx_m, x_ny_m, x_nz_m, y_nx_m, y_ny_m, y_nz_m, w_nx_m, w_ny_m, w_nz_m\},\]
then $\omega$ would contain $2$ additional palindromic factors: $bbaaabb$, $abbaaabba.$  So, either $\#\Pal(\omega)\geq 13$ or none of the above combinations occurs in $\omega$. But if none of the above combinations occurs in $\omega$, since $\omega$ is closed under reversal, this implies that $w_n$ does not occur in $\omega$ (since $w_n$ can only be followed by $w_m)$, and hence neither does $z_n.$ Since $x_n$ and $y_n$ can only be followed by $w_m,$ it follows that $x_n$ and $y_n$ also do not occur in $\omega$. Since $\omega$ is recurrent, some first return to $baaab$ must occur. Thus, even in this case we have $\#\Pal(\omega)\geq 13.$ This completes Case 1.3.\\

\noindent {\it Case 1.4:} $aaaaa \in \Pal(\omega).$ In this case
\[\{\varepsilon, a, b, aa, bb, aaa, aba, bab, aaaa, abba, baab, aaaaa, baaab, baaaab, baaaaab\}\subseteq \Pal(\omega),\] whence $\#\Pal(\omega)\geq 15.$ This completes Case 1.\\

\noindent {\it Case 2:} $\omega$ does not contain any palindromic factors of length $5.$\\

In this case, neither $aaa$ nor $bbb$ is a factor of $\omega$ (for otherwise by Lemma~\ref{bab} $\omega$ would contain a palindrome of length $5).$
But in view of Proposition~\ref{5&6}, $\omega$ must contain a palindromic factor of length $6.$ Without loss of generality this implies that 
\[\{aabbaa, babbab\}\cap \Pal(\omega) \neq \emptyset.\]

\noindent {\it Case 2.1:} $aabbaa \in \Pal(\omega).$ In this case, $baabbaab \in \Pal(\omega)$ so that
$\#\Pal(\omega)\geq 11.$ We now consider the possible bilateral extensions of $baabbaab.$ 
The extension $abaabbaaba$ gives rise to $2$ additional palindromes: $abaabbaaba$, $babaabbaabab$ where the second follows from the fact that $\omega$ does not contain any palindromic factors of length $5.$
The extension $abaabbaabb$ gives rise to $2$ additional palindromes: $bbaabb$, $abbaabba.$ Finally, the last extension $bbaabbaabb$ gives rise to $2$ additional palindromes: $bbaabb$, $bbaabbaabb.$ In either case, $\#\Pal(\omega)\geq 13.$\\

\noindent {\it Case 2.2:} $babbab \in \Pal(\omega).$ Since $\omega$ contains no palindromic factor of length $5$, the only bilateral extension of $babbab$ is $ababbaba.$ So $\#\Pal(\omega)\geq 11.$ Again, since $\omega$ contains no palindromic factor of length $5,$ the only bilateral extension of $ababbaba$ is $aababbabaa$, which gives rise to $2$ additional palindromes: $aababbabaa$, $baababbabaab.$ Thus, $\#\Pal(\omega)\geq 13.$ This completes Case 2 and the proof of Theorem~\ref{13}.
\end{proof}

\begin{remark}
 One can wonder what happens for infinite words that are generated by morphisms. Actually, Tan \cite{BoTan} proved that if $\omega $ is the fixed point of a primitive morphism, then $\omega $ is closed under reversal if and only if $\#\Pal(\omega)=\infty.$
\end{remark}

%%%%%%%%%%%%%%%%%%%%%%%%%%%%%%%%%%%%%%%%%%%%%%%%%%%%%%%%%%%%%%%%%%%%%%%%%%%%%%%%%%%%%%%%%%%%%%
\section{Acknowledgements}
%%%%%%%%%%%%%%%%%%%%%%%%%%%%%%%%%%%%%%%%%%%%%%%%%%%%%%%%%%%%%%%%%%%%%%%%%%%%%%%%%%%%%%%%%%%%%%

The second author is partially supported by a FiDiPro grant (137991) from the Academy of Finland, and by ANR grant {\sl SUBTILE}. 

%%%%%%%%%%%%%%%%%%%%%%%%%%%%%%%%%%%%%%%%%%%%%%%%%%%%%%%%%%%%%%%%%%%%%%%%%%%%%%%%%%%%%%%%%%%%%%


\begin{thebibliography}{9999999}

\bibitem{AB1} B.~Adamczewski, Y.~Bugeaud, Palindromic continued fractions, {\it Ann. Inst. Fourier} {\bf 57} (2007) p. 1557--1574.

\bibitem{AB2} B.~Adamczewski, Y.~Bugeaud, Transcendence measure for continued fractions involving repetitive or symmetric patterns, {\it J. Eur. Math. Soc.} {\bf 12} (2010) p. 883--914.

\bibitem{AB3} B.~Adamczewski, Y.~Bugeaud, On the Littlewood conjecture in simultaneous Diophantine
approximation, {\it J. London Math. Soc.} {\bf 73} (2006) p. 355--366. 

\bibitem{AA} B.~Adamczewski, J.-P.~Allouche, Reversal and palindromes in continued fractions, {\it Theoret. Comput. Sci.} {\bf 320} (2007) 220--237. 

\bibitem{All} J.-P.~Allouche, Schr\"{o}dinger operators with Rudin-Shapiro potentials are not palindromic, {\it J. Math. Phys., Special Issue ``Quantum Problems in Condensed Matter Physics''} {\bf 38} (1997) p. 1843--1848.

\bibitem{jAmBjCdD03pali} J.-P.~Allouche, M.~Baake, J.~Cassaigne, D.~Damanik, Palindrome complexity, {\it Theoret. Comput. Sci.} {\bf 292} (2003) p. 9--31.

\bibitem{pAzMePcF06pali} P.~Ambro{\v{z}}, C.~Frougny, Z.~Mas{\'a}kov{\'a}, E.~Pelantov{\'a},  Palindromic complexity of infinite words associated with 
simple {P}arry numbers, {\it Ann. Inst. Fourier (Grenoble)} {\bf 56} (2006) p. 2131--2160.
  
\bibitem{pBzMeP07fact} P.~Bal\'a\v{z}i, Z.~Mas\'akov\'a, E.~Pelantov\'a, Factor versus
palindromic complexity of uniformly recurrent infinite words, {\it Theoret. 
Comput. Sci.} {\bf 380} (2007) p. 266--275. 

\bibitem{Ber} J.~Berstel, L.~Boasson, O.~Carton and I.~Fagnot, Infinite words without palindrome, {\it CoRR abs/0903.2382}, (2009).

\bibitem{Brl} S.~Brlek, S.~Hamel, M.~Nivat and C.~Reutenauer. On The Palindromic Complexity Of Infinite Words, {\it Internat. J. Found. Comput. Sci.} {\bf 15} (2004) p. 293--306.

\bibitem{BDGZ1} M.~Bucci, A.~De Luca, A.~Glen, L.Q.~Zamboni, A connection between palindromic and factor complexity using return words, {\it Adv. Appl. Math.} {\bf 42} (2009) p. 60--74.

\bibitem{BDGZ2} M.~Bucci, A.~De~Luca, A.~Glen, L.Q.~Zamboni, A new characteristic property of rich words, {\it Theoret. Comput. Sci.} {\bf 410} (2009) p. 2860--2863. 

\bibitem{BL} Y.~Bugeaud,  M.~Laurent, Exponents of Diophantine and Sturmian continued fractions,
{\it Ann. Inst. Fourier} {\bf 55} (2005) p. 773--804.

\bibitem{dDlZ03comb}
D.~Damanik, L.Q.~Zamboni, Combinatorial properties of {A}rnoux-{R}auzy subshifts and applications to {S}chr\"odinger operators, {\it Rev. Math. Phys.} {\bf 15} (2003) p. 745--763.

\bibitem{aD97stur} A.~de~Luca, Sturmian words: structure,
combinatorics and their arithmetics, \emph{Theoret. Comput. Sci.}
{183} (1997) 45--82. 

\bibitem{DLGZ} A.~de~Luca, A.~Glen, L.Q.~Zamboni, Rich, Sturmian, and trapezoidal words, {\it Theoret. Comput. Sci.} {\bf 407} (2008) p. 569--573.

\bibitem{xDgP99pali}
X.~Droubay, G.~Pirillo, Palindromes and {S}turmian words, {\it Theoret. 
  Comput. Sci.} {\bf 223} (1999) p. 73--85. 

\bibitem{xDjJgP01epis} X.~Droubay, J.~Justin, G.~Pirillo,
Episturmian words and some constructions of de Luca and Rauzy,
{\it Theoret. Comput. Sci.} {\bf 255} (2001) p. 539--553. 

\bibitem{sF06pali2}
S.~Fischler, Palindromic prefixes and diophantine approximation, {\it Monatsh. Math.} {\bf 151} (2007) p. 11--37.

\bibitem{GJWZ} A.~Glen, J.~Justin, S.~Widmer, L.Q.~Zamboni, Palindromic richness,  {\it European J. Combinatorics}, {\bf 30} (2009) p. 510--531.

\bibitem{aHoKbS95sing} A.~Hof, O.~Knill, B.~Simon, Singular continuous spectrum for palindromic 
Schr\"{o}dinger operators, {\it Commun. Math. Phys.} {\bf 174} (1995) p. 149--159.

\bibitem{mL83comb} M.~Lothaire, {\em Combinatorics On Words}, vol.~17 of {\em Encyclopedia of Mathematics and its Applications}, Addison-Wesley, Reading, Massachusetts, 1983.

\bibitem{mL02alge} M.~Lothaire, \emph{Algebraic Combinatorics On
Words}, vol.~90 of {\em Encyclopedia of Mathematics and its Applications}, Cambridge University Press, U.K., 2002.
 
\bibitem{Roy1} D.~Roy, Approximation to real numbers by cubic algebraic integers, II, {\it Ann. of
Math.} {\bf 158} (2003) p. 1081--1087.

\bibitem{Roy2} D.~Roy, Approximation to real numbers by cubic algebraic integers, I, {\it Proc. London
Math. Soc.}  {\bf 88} (2004) p. 42--62.

\bibitem{BoTan} B.~Tan. Mirror substitutions and palindromic sequences, {\it Theoret. Comput. Sci.} {\bf 389} (2007) p. 118--124.

\end{thebibliography}
\end{document}